\documentclass{article}

\usepackage{times}
\usepackage[table,usenames,dvipsnames]{xcolor}% http://ctan.org/pkg/xcolor
\usepackage[T1]{fontenc}
\usepackage[utf8]{inputenc}
\usepackage{authblk}
\usepackage{graphicx} 
\usepackage{listings} 
\usepackage{mathtools}
\usepackage{amssymb,amsfonts}
\usepackage{todonotes}
\usepackage[normalem]{ulem}
\usepackage[linesnumbered,ruled,vlined]{algorithm2e}
\usepackage{pgf,tikz}

\newenvironment{proof} {\noindent\emph{Proof:}}{$\left.\right.$\hfill$\square$\medskip} % proof with end of proof symbol

\usepackage[bibliography=common,appendix=append]{apxproof}
\usetikzlibrary{arrows}

\renewcommand{\appendixrefname}{Appendix}

\newtheoremrep{proposition}{Proposition}
\newtheorem{definition}{Definition}
\newtheorem{example}{Example}
\newtheorem{corollary}{Corollary}
\newtheorem{lemma}{Lemma}

\newcommand{\fun}[1]{\ensuremath{\textsl{#1}}} 

\newcommand{\body}[1]{\ensuremath{\mathrm{body}(#1)}}
\newcommand{\head}[1]{\ensuremath{\mathrm{head}(#1)}}
\newcommand{\erule}{\ensuremath{\sigma}}
\newcommand{\fr}[1]{\ensuremath{\mathrm{fr}(#1)}}
\newcommand{\ruleset}{\ensuremath{\Sigma}}
\newcommand{\instance}{\ensuremath{I}}

\newcommand{\matchsafe}{\ensuremath{\match^s}}
\newcommand{\kb}{\ensuremath{\mathcal{K}}}
\newcommand{\kblong}{\ensuremath{(\instance,\ruleset)}}
\newcommand{\predicate}{\ensuremath{r}}

\newcommand{\set}{\ensuremath{S}}

\newcommand{\vars}[1]{\fun{vars}{(#1)}}
\newcommand{\terms}[1]{\fun{terms}{(#1)}}
\newcommand{\atoms}[1]{\fun{atom}{(#1)}}	% mis au singulier pour les regles lineaires !
\newcommand{\treeatoms}[1]{\fun{atoms}(#1)}
\newcommand{\generated}[1]{\fun{generated}{(#1)}}
\newcommand{\shared}[1]{\fun{shared}{(#1)}}
\newcommand{\term}{\ensuremath{t}}

\newcommand{\type}[1]{\fun{type}{(#1)}}

\newcommand{\derivation}{\ensuremath{S}}
\newcommand{\chase}[2]{\fun{chase}(#1,#2)}
\newcommand{\atom}{\alpha}
\newcommand{\vect}[1]{\mathbf{#1}}
\newcommand{\match}{\pi}
\newcommand{\bag}{\ensuremath{B}}

\newcommand{\stequiv}{\ensuremath{\equiv_{st}}}
\newcommand{\derivationtree}{\ensuremath{\mathcal{T}}}
\newcommand{\entailmenttree}{\ensuremath{\mathcal{T}}}

\title{A Single Approach to Decide Chase Termination on Linear Existential Rules}

\author[1]{Michel Leclère}%{University of Montpellier, CNRS, Inria, LIRMM, France}{leclere@lirmm.fr}{}{}

\author[1]{Marie-Laure Mugnier}%{University of Montpellier, CNRS, Inria, LIRMM, France}{mugnier@lirmm.fr}{}{}

\author[2]{Micha\"el Thomazo}%{Inria, DI ENS, ENS, CNRS, PSL University, France}{michael.thomazo@inria.fr}{}{}

\author[1]{Federico Ulliana}%{University of Montpellier, CNRS, Inria, LIRMM, France}{ulliana@lirmm}{}{}

\affil[1]{University of Montpellier, CNRS, Inria, LIRMM, France}
\affil[2]{Inria, DI ENS, ENS, CNRS, PSL University, France}

\begin{document}

\maketitle

\begin{abstract}
Existential rules, long known as tuple-generating dependencies in database theory, have been intensively studied in the last decade as a powerful formalism to represent ontological knowledge in the context of ontology-based query answering.
A knowledge base is then composed of an instance that contains incomplete data and a set of existential rules, and answers to queries are logically entailed from the knowledge base. This brought again to light the fundamental chase tool, and its different variants that have been proposed in the literature. It is well-known that the problem of determining, given a chase variant and a set of existential rules, whether the chase will halt on any instance, is undecidable. Hence, a crucial issue is whether it becomes decidable for known subclasses of existential rules. In this work, we consider linear existential rules, a simple yet important subclass of existential rules that generalizes inclusion dependencies. We show the decidability of the \emph{all instance} chase termination problem on linear rules for three main chase variants, namely \emph{semi-oblivious}, \emph{restricted} and \emph{core} chase.  To obtain these results, we introduce a novel approach based on so-called derivation trees and a single notion of forbidden pattern. Besides the theoretical interest of a unified approach and new proofs, we provide the first positive decidability results concerning the termination of the restricted chase, proving that chase termination on linear existential rules is decidable for both versions of the problem: Does \emph{every} fair chase sequence terminate? Does \emph{some} fair chase sequence terminate? 

\end{abstract}

\section{Introduction}

The chase procedure is a fundamental tool for solving many issues involving tuple-generating dependencies, such as data integration \cite{DBLP:conf/pods/Lenzerini02}, data-exchange \cite{DBLP:journals/tcs/FaginKMP05}, query answering using views \cite{DBLP:journals/vldb/Halevy01} or query answering on probabilistic databases \cite{DBLP:conf/icde/OlteanuHK09}. In the last decade, tuple-generating dependencies raised a renewed interest under the name of \emph{existential rules} for the problem known as ontology-based query answering. In this context, the aim is to query a knowledge base $(\instance, \ruleset)$, where $\instance$ is an instance and $\ruleset$ is a set of existential rules (see e.g. the survey chapters \cite{DBLP:books/sp/virgilio09/CaliGL09,DBLP:conf/rweb/MugnierT14}). In more classical database terms, this problem can be recast as querying an instance $\instance$ under incomplete data assumption, provided with a set of constraints $\ruleset$, which are tuple-generating dependencies. The chase is a fundamental tool to solve dependency-related problems as it allows one to compute a (possibly infinite) \emph{universal model} of $(\instance, \ruleset)$, \emph{i.e.}, a model that can be homomorphically mapped to any other model of $(\instance, \ruleset)$. Hence, the answers to a conjunctive query (and more generally to any kind of query closed by homomorphism) over $(\instance, \ruleset)$ can be defined by considering solely this universal model. 
  
Several variants of the chase have been introduced, and we focus in this paper on the main ones: semi-oblivious  \cite{DBLP:conf/pods/Marnette09} (aka skolem \cite{DBLP:conf/pods/Marnette09}), restricted  \cite{DBLP:conf/icalp/BeeriV81,DBLP:journals/tcs/FaginKMP05} (aka standard \cite{phd/Onet12}) and core  \cite{DBLP:conf/pods/DeutschNR08}. It is well known that all of these produce homomorphically equivalent results but terminate for increasingly larger subclasses of existential rules. 

Any chase variant starts from an instance and
exhaustively performs a sequence of rule applications according to a redundancy criterion which characterizes the variant itself. 
The question of whether a chase variant terminates on \emph{all instances} for a given set of existential rules is known to be undecidable when there is no restriction on the kind of rules  \cite{DBLP:journals/ai/BagetLMS11,DBLP:conf/icalp/GogaczM14}. A number of  \emph{sufficient} syntactic conditions for termination have been proposed in the literature for the  semi-oblivious chase (see e.g. \cite{phd/Onet12,DBLP:journals/jair/GrauHKKMMW13,DBLP:phd/hal/Rocher16} for syntheses), as well as for the restricted chase \cite{DBLP:conf/ijcai/CarralDK17} (note that the latter paper also defines a sufficient condition for non-termination).  However, only few positive results exist regarding the termination of the chase on specific classes of rules. Decidability was shown for the semi-oblivious chase on guarded-based rules (linear rules, and their extension to (weakly-)guarded rules) \cite{DBLP:conf/pods/CalauttiGP15}. Decidability of the core chase termination on guarded rules for a fixed instance was shown in \cite{DBLP:conf/icdt/Hernich12}.  

In this work, we provide new insights on the chase termination problem for \emph{linear} existential rules, a simple yet important subclass of guarded existential rules, which generalizes inclusion dependencies
\cite{DBLP:journals/tods/Fagin81} and practical ontological languages   \cite{DBLP:journals/ws/CaliGL12}. Precisely, the question of whether a chase variant terminates on all instances for a set of linear existential rules is studied in two fashions: 
\begin{itemize}
\item does \emph{every} (fair)  chase sequence terminate? 
\item  does \emph{some} (fair) chase sequence  terminate? 
\end{itemize}
It is well-known that these two questions have the same answer for the semi-oblivious and the core chase variants, but not for the restricted chase. Indeed, this last one  may admit both terminating and non-terminating sequences over the same knowledge base.
We show that the termination problem is decidable for linear existential rules, 
whether we consider any version of the problem and any chase variant.

We study chase termination by exploiting in a novel way a graph structure, 
namely the \emph{derivation tree}, which was originally introduced to solve the ontology-based (conjunctive) query answering problem for the family of greedy-bounded treewidth sets of existential rules \cite{DBLP:conf/ijcai/BagetMRT11,DBLP:phd/hal/Thomazo13}, a class that generalizes guarded-based rules and  in particular linear  rules. 
We first use derivation trees to show the decidability of the termination problem for the semi-oblivious and restricted chase variants, and then  generalize them to \emph{entailment trees} to show the decidability of  termination for the core chase. For any chase variant we consider, we adopt the same high-level procedure: starting from a finite set of canonical instances (representative of all possible instances), we build a (set of) tree structures for each canonical instance, while forbidding the occurrence of a specific pattern, 
we call \emph{unbounded-path witness}. 
The built structures are finite thanks to this forbidden pattern,
and this allows us to decide if the chase terminates on the associated canonical instance. 
By doing so, we obtain a uniform approach to study the termination of several chase variants, that we believe to be of theoretical interest per se.  The derivation tree is moreover a simple structure and the algorithms built on it are likely to lead to an effective implementation. Let us also point out that our approach is constructive: if the chase terminates on a given instance, the algorithm that decides termination actually computes the result of the chase (or a superset of it in the case of the core chase), otherwise it pinpoints a forbidden pattern responsible for non-termination. 

Besides providing new theoretical tools to study chase termination, we obtain the following results for linear existential rules:
\begin{itemize}
\item a new proof of the decidability of the semi-oblivious chase termination, building on different objects than the previous proof provided in \cite{DBLP:conf/pods/CalauttiGP15}; we show that our algorithm provides the same complexity upper-bound;
\item the decidability of the restricted chase termination, for both versions of the problem, i.e., termination of all (fair) chase sequences and termination of some (fair) chase sequence; to the best of our knowledge, these are the first positive results on the decidability of the restricted chase termination;  
\item a new proof of the decidability of the core chase termination, with different objects than previous work reported in \cite{DBLP:conf/icdt/Hernich12}; although this latter paper solves the question of the core chase termination given a \emph{single} instance, the results actually allow to infer the decidability of the \emph{all} instance version of the problem, by noticing that only a finite number of  instances need to be considered (see the next section). 
\end{itemize}

The paper is organized as follows. After introducing some preliminary notions (Section 2), we define the main components of our framework, namely derivation trees and unbounded-path witnesses (Section 3). 
We build on these objects to prove the decidability of the semi-oblivious and restricted chase termination (Section 4). Finally, we generalize derivation-trees to  entailment trees and use them to prove the decidability of the core chase termination (Section 5).  Detailed proofs are provided in the appendix. 

%%%%%%%%%%%%%%%%%

\section{Preliminaries}
\label{section-preliminaries}
We consider a logical \emph{vocabulary} composed of a finite set of predicates and an infinite set of constants.
An \emph{atom} $\atom$ has the form $\predicate(t_1,\ldots, t_n)$ where $\predicate$ is a predicate of arity $n$ and the $t_i$ are terms (i.e., variables or constants).  
We denote by $\terms{\atom}$ 
(resp. $\vars{\atom}$) the set of terms 
(resp. variables) in $\atom$ and extend the notations to a set of atoms. A \emph{ground} atom does not contain any variable. 
 It is convenient to identify the existential closure of a conjunction of atoms with the set of these atoms. 
An \emph{instance} is a set of (non-necessarily ground) atoms, which is finite unless otherwise specified. 
Abusing terminology, we will often see an instance as its isomorphic model.

Given two sets of atoms $\set$ and $\set'$,  a \emph{homomorphism} from $\set'$ to $\set$ is a substitution $\match$ of $\vars{\set'}$ by $\terms{\set}$ such that $\match(\set') \subseteq \set$. It holds that $\set \models \set'$ (where $\models$ denotes classical logical entailment) iff there is a homomorphism from $\set'$ to $\set$. An endomorphism of $\set$ is a homomorphism from $\set$ to itself. 
A set of atoms is a \emph{core} if it admits only injective endomorphisms. Any finite set of atoms is logically equivalent to one of its subsets that is a core, and this core is unique up to isomomorphism (i.e., bijective variable renaming). 
Given sets of atoms $\set$ and $\set'$ such that $\set \cap \set' \neq \emptyset$, we say that $\set$ \emph{folds} onto $\set'$ if there is a homomorphism $\match$ from $\set$ to $\set'$ such that $\match$ is the identity on $\set \cap \set'$. The homomorphism $\match$ is called a \emph{folding}. In particular, it is well-known that any set of atoms \emph{folds} onto its core.

An existential rule (or simply \emph{rule}) is of  the form 
$\erule = \forall\vect{x}\forall\vect{y}.[\body{\vect{x},\vect{y}} \rightarrow \exists \vect{z}.\head{\vect{x},\vect{z}}]$
where $\body{\vect{x},\vect{y}}$ and $\head{\vect{x},\vect{z}}$ are non-empty conjunctions of atoms on variables, respectively called the \emph{body} and  the \emph{head} of the rule, also denoted by 
$\body{\erule}$ and $\head{\erule}$, and $\vect{x}, \vect{y}$ and $\vect{z}$ are pairwise disjoint tuples of variables. The variables of $\vect z$ are called \emph{existential variables}. The variables of $\vect x$ form the \emph{frontier} of $\erule$, which is also denoted by $\fr{\erule}$. 
For brevity, we will omit universal quantifiers in the examples. 
A \emph{knowledge base} (KB) is of the form $\kb = \kblong$, where $\instance$ is an instance
and $\ruleset$ is a finite set of existential rules.  
 
 A  rule $\erule = \body{\erule} \rightarrow \head{\erule}$ is \emph{applicable} to an instance $\instance$ if there is a homomorphism $\match$ from $ \body{\erule} $ to $\instance$.  The pair $(\erule, \match)$ is called a \emph{trigger} for $\instance$. The result of the application of $\erule$ according to $\match$ on $\instance$ is the instance $\instance' = \instance \cup \matchsafe(\head{\erule})$, where $\matchsafe$ (here $s$ stands for \emph{safe}) extends $\match$ by assigning a distinct fresh variable (also called a \emph{null}) to each existential variable. We also say that  $\instance'$ is obtained by \emph{firing} the trigger $(\erule, \match)$ on $\instance$. 
 By $\match_{\mid\fr{\erule}}$ we denote the restriction of  $\match$ to the domain $\fr{\erule}$.

\sloppy
\begin{definition}[Derivation]
A \emph{$\ruleset$-derivation} (or simply \emph{derivation} when $\ruleset$ is clear from the context) 
from an instance $\instance = I_0$ to an instance $I_n$ is a sequence 
$I_0, (\erule_1,\match_1), I_1 \ldots, I_{n-1}, (\erule_n,\match_n), I_n$, such that for all $1 \leq i \leq n$: 
$\erule_i \in \ruleset$, $(\erule_i,\match_i)$ is a trigger for $I_{i-1}$, $I_i$ is obtained by firing $(\erule_i,\match_i)$ on $I_{i-1}$,
and $I_i \neq I_{i-1}$.  We may also denote this derivation by the associated sequence of instances $(I_0, \ldots, I_n)$ when the triggers are not needed. 
The notion of derivation can be naturally extended to an \emph{infinite} sequence.
\end{definition}
\fussy

We briefly introduce below the main chase variants and refer to  \cite{phd/Onet12} for a detailed presentation.

The \emph{semi-oblivious} chase prevents several applications of the same rule through the same mapping of its frontier. Given a derivation from $I_0$ to $I_{i}$, a trigger $(\erule,\match)$ for $I_i$ is said to be \emph{active according to the semi-oblivious criterion}, if there is no trigger $(\erule_j,\match_j)$ in the derivation with $\erule = \erule_j$ and   $\match_{\mid{\fr{\erule}}} = \match_{j_\mid{\fr{\erule_j}}}$. The \emph{restricted} chase performs a rule application only if the added set of atoms is not redundant with respect to the current instance.  Given a derivation from $I_0$ to $I_{i}$, a trigger $(\erule,\match)$ for $I_i$ is said to be \emph{active according to the restricted criterion}
if $\match$ cannot be extended to a homomorphism from $(\body{\erule}\cup\head{\erule})$ to $\instance_{i}$ (equivalently, $ \match^s(\head{\erule})$ does not fold onto $\instance_{i}$). A 
 \emph{semi-oblivious (resp. restricted) chase sequence} of $\instance$ with $\ruleset$ is a possibly infinite $\ruleset$-derivation from $\instance$ such that each trigger $(\erule_i,\match_i)$ in the derivation is active according to the semi-oblivious (resp. restricted) criterion. 
  
  Furthermore, a (possibly infinite) chase sequence is required to be \emph{fair}, which means that a possible rule application is not indefinitely delayed. Formally, if some $I_i$ in the derivation admits an active  trigger $(\erule,\match)$, then there is  $j > i$ such that, either $I_j$ is obtained by firing $(\erule,\match)$ on $I_{j-1}$, or $(\erule,\match)$ is not an active trigger anymore on $I_j$.  
A \emph{terminating} chase sequence is a finite fair sequence.  

In its original definition {\cite{DBLP:conf/pods/DeutschNR08}, the \emph{core} chase proceeds in a breadth-first manner, and, at each step, first fires in parallel all active triggers according to the restricted chase criterion, then computes the core of the result. Alternatively, to bring the definition of the core chase closer to the above definitions of the semi-oblivious and restricted chases,} one can define a \emph{core chase sequence} as a possibly infinite sequence $I_0, (\erule_1,\match_1), I_1, \ldots$, alternating instances and triggers, such that each instance $I_i$ is obtained from $I_{i-1}$ by first firing the active trigger $(\erule_i, \match_i)$ according to the restricted criterion, then computing the core of the result. An instance admits a terminating core chase sequence in that sense if and only if the core chase as originally defined terminates on that instance. 

For the three chase variants, fair chase sequences compute a (possibly infinite) \emph{universal model} of the KB, 
but only the core chase stops if and only if the KB has a \emph{finite} universal model. 

\medskip
It is well-known that, for the semi-oblivious and the core chase, if there is a terminating chase sequence from  an instance $I$ then all fair sequences from $I$ are terminating. 
This is not the case for the restricted chase, since the order in which rules are applied has an impact on termination, as illustrated by Example \ref{ex-intro}. 

\begin{example}\label{ex-intro} Let $\ruleset = \{\erule_1, \erule_2\}$, with  $\erule_1 = p(x,y) \rightarrow \exists z ~ p(y,z)$ and $\erule_2 =  p(x,y) \rightarrow p(y,y)$.
Let $\instance = p(a,b)$. The KB $(\instance, \ruleset)$ has a finite universal model, for example, $I^* = \{p(a,b), p(b,b)\}$. 
The semi-oblivious chase does not terminate on $\instance$ as  $\erule_1$ is applied indefinitely, while the core chase terminates after one breadth-first
 step and returns $I^*$.  The restricted chase has a terminating sequence, for example, $(\erule_2, \{x \mapsto a, y \mapsto b\})$, which yields $I^*$ as well, but it also has infinite fair sequences,
for example, the breadth-first sequence that applies $\erule_1$ before $\erule_2$ at each step. 
\end{example}

\pagebreak
We study the following problems for the semi-oblivious, restricted and core chase variants: 
\begin{itemize}
\item \emph{(All instance) all sequence termination:} Given a set of rules $\ruleset$, is it true that, for any instance,  all fair sequences are terminating?  
\item \emph{(All instance) one sequence termination:} Given a set of rules $\ruleset$, is it true that, for any instance,  there is a terminating sequence?
\end{itemize}

Note that, according to the terminology of \cite{DBLP:journals/fuin/GrahneO18}, these problems can be recast as deciding whether, for a  chase variant, a given set of rules belongs to the class  CT$_{\forall\forall}$ or  CT$_{\forall\exists}$, respectively. 

\medskip
An existential rule is called \emph{linear} if its body and its head are both composed of a single atom (e.g., \cite{DBLP:journals/ws/CaliGL12}). Linear rules generalize \emph{inclusion dependencies} 
\cite{DBLP:journals/tods/Fagin81} by allowing several occurrences of the same variable in an atom. They also generalize positive inclusions in the description logic DL-Lite$_\mathcal R$  (the formal basis of the web ontological language OWL2 QL) \cite{DBLP:journals/jar/CalvaneseGLLR07}, which can be seen as inclusion dependencies restricted to unary and binary predicates. 

Note that the restriction of existential rules to rules with a single head is often made in the literature, considering that any existential rule with a complex head can be decomposed into several rules with a single head, by introducing a fresh predicate for each rule. However, while this translation preserves the termination of the semi-oblivious chase, it is not the case for the restricted and the core chases. Hence, considering linear rules with a complex head would require to extend the techniques developed in this paper.

\medskip
To simplify the presentation, we assume in the following that each rule frontier is of size at least one. This assumption is made without loss of generality. \footnote{For instance, it can always be ensured by adding a position to all predicates, which is filled by the same fresh constant in the initial instance, and by a new frontier variable in each rule.} 

\medskip
We first point out that the termination problem on linear rules can be recast by considering solely instances that contain a single atom (as already remarked in several contexts). 

\begin{propositionrep}
\label{prop-atomic-instance}
 Let $\ruleset$ be a linear set of rules. The semi-oblivious (resp. restricted, core) chase terminates on all instances if and only if it terminates on all singleton instances.
\end{propositionrep}
\begin{proof}
Obviously, the fact that a chase variant does not halt on  an atomic instance implies the fact that it does not terminate on all instances. On the other direction, we can easily see that if the chase does not halt on an instance then it will not halt on one of its atoms.
For a chase variant that does not terminate there exists an infinite derivation whose associated chase graph is also infinite.
As the arity of the nodes in the chase graph is bounded by the size of the ruleset, the chase graph must contains an infinite path starting from a node of the initial instance.
Because the chase graph for linear rules forms a tree it follows that this infinite path is created by a single atom of the initial instance.

\end{proof}

We will furthermore rely on the following notion of the type of an atom. 
 
  \begin{definition}[Type of an atom]
\label{definition-type}
 The \emph{type of an atom} $\atom = r(t_1,\ldots, t_n)$, denoted by $\type{\atom}$, is the pair $(r,\mathcal{P})$ where $\mathcal{P}$ is the partition of $\{1,\ldots,n\}$ induced by term equality (i.e., $i$ and $j$ are in the same class of $\mathcal{P}$ iff $\term_i = \term_j$).
 \end{definition} 
 
Note that there are finitely (more specifically, exponentially) many types for a given vocabulary. 

If two atoms $\atom$ and $\atom'$ have the same type, then there is a \emph{natural mapping} from $\atom$ to $\atom'$, denoted by $\varphi_{\atom\rightarrow \atom'}$, and defined as follows:
it is a bijective mapping from $\terms{\atom}$ to $\terms{\atom'}$, that maps the $i$-th term of $\atom$ to the $i$-th term of $\atom'$. 
Note that $\varphi_{\atom\rightarrow \atom'}$ may not be an isomorphism, as constants from $\atom$ may not be mapped to themselves. However, if $(\erule,\match)$ is a trigger for $\{\atom\}$, then $(\erule,\varphi_{\atom\rightarrow \atom'}\circ\match)$ is a trigger for $\{\atom'\}$, as there are no constants in the considered rules. 

Together with Proposition \ref{prop-atomic-instance}, this implies that one can check all instance all sequence termination by checking all sequence termination on a finite set of instances, called \emph{canonical instances}: for each type, there is exactly one canonical instance that has this type.

We will consider different kinds of tree structures, which have in common to be \emph{trees of bags}: these are rooted trees,  whose nodes, called \emph{bags}, are labeled by an atom.\footnote{Furthermore the trees we will consider are decomposition trees of the associated set of atoms. That is why we use the classical term of \emph{bag} to denote a node.}
We define the following notations for any node $\bag$ of a tree of bags $\derivationtree$:
\begin{itemize}
\item  $\atoms{\bag}$ is the label of $\bag$;
\item $\terms{\bag} = \terms{\atoms{\bag}}$ is the set of terms of $B$;
\item  $\terms{\bag}$ is divided into two sets of terms, those \emph{generated} in $\bag$, denoted by $\generated{\bag}$, and those shared with its parent, denoted by 
$\shared{\bag}$; precisely, $\terms{\bag} = \shared{\bag} \cup \generated{\bag}$,  $\shared{\bag} \cap \generated{\bag} = \emptyset$,  and if $\bag$ is the root of $\derivationtree$, then $\generated{\bag} = \terms{\bag}$ (hence $\shared{\bag} = \emptyset$), otherwise $\bag$ has a parent $\bag_p$ and $\generated{\bag} = \terms{\bag} \setminus \terms{\bag_p}$ (hence, $\shared{\bag} = \terms{\bag_p} \cap \terms{\bag}$).
\end{itemize}

 We denote by $\treeatoms{\derivationtree}$ the set of atoms that label the bags in $\derivationtree$.

Finally, we recall some classical mathematical notions. A \emph{subsequence} $S'$ of a sequence  $S$ is a sequence that can be obtained from $S$ by deleting some (or no) elements without changing the order of the remaining elements. The \emph{arity} of a tree is the maximal number of children for a node. A \emph{prefix} $T'$ of a tree $T$ is a tree that can be obtained from $T$ by repeatedly deleting some (or no) leaves of $T$. 

\section{Derivation Trees}

A classical  tool to reason about the chase is the so-called \emph{chase graph} (see e.g., \cite{DBLP:journals/ws/CaliGL12}), which is the directed graph consisting of all atoms that appear in the considered derivation, and with an arrow from a node $n_1$ to a node $n_2$ iff $n_2$ is created by a rule application on $n_1$ and possibly other atoms.   \footnote{Note that the chase graph in \cite{DBLP:conf/pods/DeutschNR08} is a different notion.} In the specific case of KBs of the form $(\{\atom\}, \ruleset)$, where $\atom$ is an atom and $\ruleset$ is a set of linear rules, the chase graph is a tree.  We recall below its definition in this specific case, in order to emphasize its differences with another tree, called \emph{derivation tree}, on which we will actually rely.

\begin{definition}[Chase Graph for Linear Rules]
 \label{definition-chase-graph}
Let $\instance$ be a singleton instance, $\ruleset$ be a set of linear rules and $\instance = I_0, (\erule_1,\match_1), I_1 \ldots, I_{n-1}, (\erule_n,\match_n), I_n$ be a semi-oblivious $\ruleset$-derivation from $\instance$. The \emph{chase graph} (also called \emph{chase tree}) assigned to $S$ is a tree of bags built as follows:
 \begin{itemize}
 \item the set of bags is in bijection with $\instance_n$ via the labeling function $\atoms{}$; 
\item the set of edges is in bijection with the set of triggers in $S$ and is built as follows: for each trigger $(\erule_i,\match_i)$  in $S$, 
there is an edge $(\bag,\bag')$ with $\atoms{\bag} = {\match_i(\body{\erule_i})}$ and $\atoms{\bag'} = \match_i^s(\head{\erule_i})$. 
\end{itemize}
\end{definition}

\begin{example}
 \label{example-chase-graph-not-enough}
 Let $\instance = q(a)$ and $\ruleset=\{\erule_1,\erule_2\}$ where
% \sloppypar{
% \smallskip
%\begin{tabular}{ll}
 $\erule_1 = q(x) \rightarrow \exists y \exists z \exists t ~p(x,y,z,t)$ 
% &
% \hspace{1cm} 
and
 $\erule_2 = p(x,y,z,t) \rightarrow p(x,z,t,y)$.
% \\
% \end{tabular}
% \smallskip
 Let $S = \instance,(\erule_1,\match_1),\instance_1,(\erule_2,\match_2),\instance_3,(\erule_2,\match_3),\instance_3$ with $\match_1 = \{ x \mapsto a\}$, 
 $\match_1^s(\head{\erule_1}) = p(a,y_0,z_0,t_0)$, $\match_2 = \{ x \mapsto a, y \mapsto y_0, z \mapsto z_0, t \mapsto t_0\}$
 and ${\match_3 = \{ x \mapsto a, y \mapsto z_0,}$ $ z \mapsto t_0, t \mapsto y_0\}$. 
 The chase graph associated with $S$ is a path of four nodes %$B_0 = q(a), B_1 = p(a,y_0,z_0,t_0)$, $B_2 = p(a,z_0,t_0,y_0)$ and $B_3 = p(a,t_0,y_0,z_0)$, 
  as represented in Figure~\ref{figure-chase-graph-not-enough}.
 %}
 %  =_{\ell}
 \end{example}

\begin{figure}[t]
\begin{center}
\tikzset{every picture/.style={line width=0.75pt}} %set default line width to 0.75pt        

\begin{tikzpicture}[x=0.75pt,y=0.75pt,yscale=-1,xscale=1,thick,scale=0.6, every node/.style={scale=0.7}]

\draw    (126, 38) circle [x radius= 35, y radius= 20]  ;
\draw    (126.25, 104.5) circle [x radius= 55.25, y radius= 24.5]  ;
\draw    (126.25, 175.5) circle [x radius= 55.25, y radius= 24.5]  ;
\draw    (126.25, 246.5) circle [x radius= 55.25, y radius= 24.5]  ;
\draw    (126,58) -- (126,80) ;

\draw    (126,129) -- (126,151) ;

\draw    (126,200) -- (126,222) ;

\draw    (366, 37) circle [x radius= 35, y radius= 20]  ;
\draw    (366.25, 103.5) circle [x radius= 55.25, y radius= 24.5]  ;
\draw    (293.25, 176.5) circle [x radius= 55.25, y radius= 24.5]  ;
\draw    (434.25, 176.5) circle [x radius= 55.25, y radius= 24.5]  ;
\draw    (366,57) -- (366.25,79) ;

\draw    (366.25,128) -- (293.25,152) ;

\draw    (366.25,128) -- (434.25,152) ;

\draw (126,37) node   {$q( a)$};
\draw (124,105) node   {$p( a,y_{0} ,z_{0} ,t_{0})$};
\draw (125,176) node   {$p( a,z_{0} ,t_{0} ,y_{0})$};
\draw (126,247) node   {$p( a,t_{0} ,y_{0} ,z_{0})$};
\draw (76,38) node   {$B_{0}$};
\draw (56,105) node   {$B_{1}$};
\draw (54,175) node   {$B_{2}$};
\draw (57,250) node   {$B_{3}$};
\draw (117,300) node  [align=left] {Chase Graph};
\draw (367,36) node   {$q( a)$};
\draw (365,104) node   {$p( a,y_{0} ,z_{0} ,t_{0})$};
\draw (292,177) node   {$p( a,z_{0} ,t_{0} ,y_{0})$};
\draw (433,177) node   {$p( a,t_{0} ,y_{0} ,z_{0})$};
\draw (317,37) node   {$B_{0}$};
\draw (297,104) node   {$B_{1}$};
\draw (221,176) node   {$B_{2}$};
\draw (505,176) node   {$B_{3}$};
\draw (364,299) node  [align=left] {Derivation Tree};

\end{tikzpicture}
\end{center}
\caption{Chase Graph and Derivation Tree of Example \ref{example-chase-graph-not-enough}}
\label{figure-chase-graph-not-enough}
\end{figure}
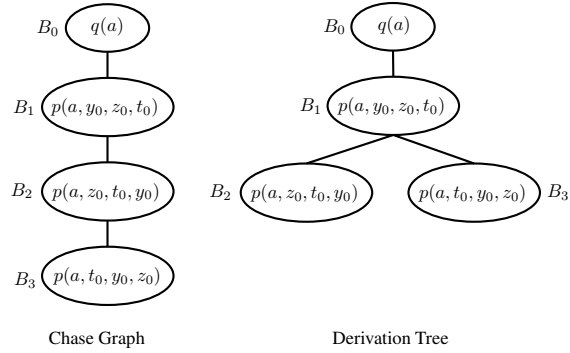

 \medskip To check termination of a chase variant on a given KB $(\{\atom\}, \ruleset)$, the general idea is to build a tree of bags associated with the chase on this KB in such a way that the occurrence of some forbidden pattern indicates that a path of unbounded length can be developed, hence the chase does not terminate.  The forbidden pattern is composed of two distinct nodes such that one is an ancestor of the other and, intuitively speaking, these nodes ``can be extended in similar ways'', which leads to an arbitrarily long path that repeats the pattern.

Two atoms with the same type admit the same rule triggers, however, within a derivation, the same rule applications cannot necessarily be performed on both of them because of the presence of other atoms (this is true already for datalog rules, since the same atom is never produced twice).  
Hence, on the one hand we will specialize the notion of type, into that of a \emph{sharing type}, and, on the other hand, adopt another tree structure,  called a \emph{derivation tree}, in which two nodes with the same sharing type have the required similar behavior. 

  \begin{definition}[Sharing type and Twins]
\label{definition-sharing-type}
Given a tree of bags,
the \emph{sharing type} of a bag $\bag$ is a pair $ (\type{\atoms{\bag}},P)$ where $P$ is the set of positions in $\atoms{\bag}$ in which a term of $\shared{\bag}$ occurs. We denote the fact that two bags $\bag$ and $\bag'$ have the same sharing type by $\bag \stequiv \bag'$. Furthermore, we say that two bags $\bag$ and $\bag'$ are \emph{twins} if they have the same sharing type, the same parent $\bag_p$ and if the natural mapping $\varphi_{\atoms{B}\rightarrow\atoms{B'}}$ is the identity on the terms of $\atoms{\bag_p}$.   
 \end{definition}
 
 We can now specify the forbidden pattern that we will consider: it is a pair of two distinct nodes with the same sharing type, such that one is an ancestor of the other. 
 
\begin{definition}[Unbounded-Path Witness]
 An \emph{unbounded-path witness} (UPW) in a derivation tree is a pair of distinct bags $(\bag,\bag')$ such that $\bag$ and $\bag'$ have the same sharing type and $\bag$ is an ancestor of $\bag'$. 
\end{definition}

 As explained below on  Example \ref{example-chase-graph-not-enough}, the chase graph is not the appropriate tool to use this forbidden pattern as a witness of chase non-termination. 
 
  \medskip
\noindent \emph{Example \ref{example-chase-graph-not-enough} (cont'd).}  $B_1$, $B_2$ and $B_3$ have the same classical type,\\
$t =  (p, \{\{1\}, \{2\},\{3\},\{4\}\})$. The sharing type of $B_1$ is $(t,\{1\})$, while $B_2$ and $B_3$ have the same sharing type 
$(t,\{1,2,3,4\})$. $B_2$ and $B_3$ fulfill the condition of the forbidden pattern, however it is easily checked that any derivation that extends this derivation is finite. 
  
\medskip 
Derivation trees were introduced as a tool to define the \emph{greedy bounded treewidth set (gbts)} family of existential rules \cite{DBLP:conf/ijcai/BagetMRT11,DBLP:phd/hal/Thomazo13}.  
A derivation tree is associated with a derivation, however it does not have the same structure as the chase graph. The fundamental reason is that, when a rule $\erule$ is applied to an atom $\atom$ via a homomorphism $\match$, the newly created bag is not necessarily attached in the tree as a child of the bag labeled by $\atom$. Instead, it is attached as a child of the \emph{highest} bag in the tree labeled by an atom that contains 
$\match(\fr{\erule})$, the image by $\match$ of the frontier of $\erule$ (note that $\match(\fr{\erule})$ remains the set of terms shared between the new bag and its parent). 

In the following definition, a derivation tree is not associated with \emph{any} derivation, but with a semi-oblivious derivation, which has the advantage  of yielding trees with bounded arity (Proposition \ref{prop-bounded-arity} in the Appendix). This is appropriate to study the termination of the semi-oblivious chase, and later the restricted chase, as a restricted chase sequence is a specific semi-oblivious chase sequence.

\begin{definition}[Derivation Tree]
 Let $\instance = \{\atom\}$ be a singleton instance, $\ruleset$ be a set of linear rules, and $\derivation = \instance_0,(\erule_1,\match_1),\instance_1, \ldots, (\erule_n,\match_n),\instance_n$ be a semi-oblivious $\ruleset$-derivation.
 The \emph{derivation tree} assigned to $\derivation$ is a tree of bags 
 $\derivationtree$ built as follows:
 \begin{itemize}
  \item the root of the tree, $\bag_0$, is such that $\atoms{\bag_0} = \atom$;
  \item for each trigger $(\erule_i,\match_i)$, $0 < i \leq n$, let $\bag_i $ be the bag such that $\atoms{\bag_i} = \matchsafe_{i}(\head{\erule_{i}})$. Let $j$ be smallest integer such that $\match_{i}(\fr{\erule_{i}}) \subseteq \terms{\bag_j}$: $\bag_i$ is added as a child to $\bag_j$. 
 \end{itemize}
 By extension, we say that a derivation tree $ \derivationtree$ is \emph{associated with $\atom$ and $\ruleset$} if there exists a semi-oblivious $\ruleset$-derivation $\derivation$ from $\atom$ such that $\derivationtree$ is assigned to $\derivation$.
  
\end{definition}

\noindent \emph{Example \ref{example-chase-graph-not-enough} (cont'd).} 
The derivation tree associated with $S$ is represented in Figure \ref{figure-chase-graph-not-enough}. Bags have the same sharing types in the chase tree and in the derivation tree. However, we can see here that they are not linked in the same way: $\bag_3$ was a child of $\bag_2$ in the chase tree, it becomes a child of $\bag_1$ in the derivation tree. Hence, the forbidden pattern cannot be found anymore in the tree. 

\medskip
Note that every non-root bag $\bag$ shares a least one term with its parent (since the rule frontiers are not empty), furthermore this term is generated in its parent (otherwise $\bag$ would have been added at a higher level in the tree).

\begin{toappendix}
\begin{proposition}\label{prop-bounded-arity}
The arity of a derivation tree is bounded. 
\end{proposition}

\begin{proof}
We first point out that a bag has a bounded number of twin children. Since we consider semi-oblivious derivations, a bag $\bag_p$ cannot have two twin children $\bag_{c_1}$ and $\bag_{c_2}$, created by  applications of the same rule $\erule$. Indeed, although these rule applications may map $\body{\erule}$ to distinct atoms, the associated homomorphisms, say $\match_1$ and $\match_2$,  would have the same restriction to the rule frontier,
i.e.,  $\match_1{_{\mid\fr{\erule}}} = \match_2{_{\mid\fr{\erule}}}$. Hence, all twin children of a bag come from applications of distinct rules. It follows that the arity of a node is bounded by the number of atom types $\times$ the cardinal of the ruleset. 
\end{proof}
\end{toappendix}

\section{Semi-Oblivious and Restricted Chase Termination}
We now use derivation trees and sharing types to characterize the termination of the semi-oblivous chase. The fundamental property of derivation trees that we exploit is that, when two nodes have the same sharing type, the considered (semi-oblivious) derivation can always be extended so that these nodes have the same number of children, and in turn these children have the same sharing type. We first specify the notion of \emph{bag copy}.

\begin{definition}[Bag Copy]
 Let $\entailmenttree,\entailmenttree'$ be two (possibly equal) trees of bags. Let $\bag$ be a bag of $\entailmenttree$ and $\bag'$ be a bag of $\entailmenttree'$  such that $\bag \stequiv \bag'$. Let $\bag_c$ be a child of $\bag$.
 A \emph{copy} of $\bag_c$ \emph{under} $\bag'$ is a bag $\bag'_c$ such that $\atoms{\bag'_c} = \varphi^s(\atoms{\bag_c})$, where $\varphi^s$ is a substitution of $\terms{\bag_c}$ defined as follows:
 \begin{itemize}
  \item if $\term \in \shared{\bag_c}$, then $\varphi^s(\term) = \varphi_{\atoms{\bag} \rightarrow \atoms{\bag'}}(\term)$, where  $\varphi_{\atoms{\bag} \rightarrow \atoms{\bag'}}$ is the natural mapping from $\atoms{\bag}$ to $\atoms{\bag'}$;
  \item if $\term \in \generated{\bag_c}$, then $\varphi^s(\term)$ is a fresh new {variable}.
 \end{itemize}
\end{definition}

Let $\derivationtree_e$ be obtained from a derivation tree $\derivationtree$ by adding a copy of a bag: strictly speaking, $\derivationtree_e$  may not be a derivation tree in the sense that there may be no derivation to which it can be assigned (intuitively, some rule applications that would allow to produce the copy may be missing). Rather, there is some derivation tree of which $\derivationtree_e$ is a \emph{prefix} (intuitively, one can add  bags to  $\derivationtree_e$ to obtain a derivation tree). That is why the following proposition considers more generally prefixes of derivation trees. 

\begin{propositionrep}
\label{proposition-sharing-type-children-derivation-tree}
 Let $\derivationtree$ be a prefix of a derivation tree, $\bag$ and $\bag'$ be two bags of $\derivationtree$ such that $\bag \stequiv \bag'$, and $\bag_c$ be a child of $\bag$. Then: \emph{(a)} the tree obtained from $\derivationtree$ by adding the copy $\bag'_c$ of $\bag_c$ under $\bag'$ is a prefix of a derivation tree, and \emph{(b)} it holds that $\bag_c \stequiv \bag'_c$.
\end{propositionrep}

\begin{proof}
 Let $\bag$ and $\bag'$ be two atoms of $\derivationtree$ having the same sharing type. 
  Let $\bag_c$ be a child of $\bag$ created by a trigger $(\erule,\match)$. 
By definition of derivation tree,  $\match$ maps the rule frontier $\fr{\erule}$ to $\terms{\bag}$, without this being possible for the parent of $\bag$. 
 Furthermore, we know that $\match$ maps  $\body{\erule}$ to  a (possibly strict) descendant of $\bag$.
 We assume that $\derivationtree$ does not already contain the image of $\head{\erule}$ via $\match$, otherwise the thesis trivially holds.
  Let $\derivation$ be the derivation associated with $\derivationtree$ and
  $\alpha_0,\dots,\alpha_{k}$ be the path of the \emph{chase-graph}  associated with $\derivation$ such that $\alpha_0=\atoms{\bag}$ and $\alpha_{k}=\atoms{\bag_{c}}$, whose sequence of associated rule applications is $(\erule_1,\match_1),\dots, (\erule_{k},\match_{k})=(\erule,\match)$. 
We define $\hat\match_i^{\mathrm{safe}}(t)=\varphi_{\atoms{\bag}\rightarrow\atoms{\bag'}}\circ\match_i(t)$ whenever $\match_i(t)\in\terms{\bag}$ and otherwise
$\hat\match_i^{\mathrm{safe}}(t)$ to be a fresh new variable consistently used over the rule applications, that is, such that 
$\match_i^{\mathrm{safe}}(t)=\match_j^{\mathrm{safe}}(t)$
if and only if 
$\hat\match_i^{\mathrm{safe}}(t)=\hat\match_j^{\mathrm{safe}}(t)$.
  Then, for all $1\leq i \leq k$, we extend $\derivation$ by adding a trigger $(\erule_i,\hat\match_i)$\footnote{$\hat\match_i$ is the restriction of $\hat\match_i^{\mathrm{safe}}$ to the variables of the body of $\erule_i$.} whenever 
  $\hat\match_i^{\mathrm{safe}}(\head{\erule_i})$ is not an atom already produced by $\derivation$ thereby obtaining a new derivation $\derivation'$. 
  Let $\derivationtree'$ be an extension of $\derivationtree$ where a bag labeled with the atom $\hat\match_i^{\mathrm{safe}}(\head{\erule_i})$ is added for each new trigger in $\derivation'$ and attached to the highest descendant of $\bag'$ whose set of terms contains $\hat\match_i^{}(\fr{\erule_i})$. Clearly, $\derivationtree'$ is a derivation tree associated with $\derivation'$. 
We now show that $\derivationtree'$ contains a node $\bag_c'$ which is a copy of $\bag_c$ under $\bag'$.
 As $\bag$ is the parent of $\bag_c$, the image of $\fr{\erule}$ via $\match$ contains at least one term which is generated in $\bag$ (and in general only terms generated by the ancestors of $\bag$).
Therefore,  because $\bag$ and $\bag'$ have the same sharing type, the image of $\fr{\erule}$ via $\varphi_{\atoms{\bag}\rightarrow \atoms{\bag'}}\circ\match$ contains at least one term generated in $\bag'$ (and in general only terms generated by the ancestors of $\bag'$). So, $\bag'$ is the only possible parent of $\bag'_c$ in $\derivationtree'$. Moreover, it is easy to see that $\bag_c \stequiv \bag'_c$. 
Let $\derivationtree''$ be the extension of $\derivationtree$ with $\bag_c'$ under $\bag'$. It can be easily verified that  $\derivationtree''$ is a prefix of the derivation tree $\derivationtree'$, in the sense that it is a tree of bags which can be obtained by recursively removing some of the leaves of $\derivationtree'$, i.e., those corresponding to the triggers in $\derivation'\setminus\derivation$ which are different from $(\erule,\match)$.
\end{proof}

The size of a derivation tree without UPW is bounded, since its arity is bounded (Proposition \ref{prop-bounded-arity} in the Appendix) and its depth is bounded by the number of sharing types. It remains to show that a derivation tree that contains a UPW can be extended to an arbitrarily large derivation tree. We recall that similar property would not hold for the chase tree, as witnessed by Example~\ref{example-chase-graph-not-enough}.

\begin{propositionrep}
\label{proposition-finiteness-derivation-tree}
 There exists an arbitrary large derivation tree associated with $\atom$ and $\ruleset$ if and only if there exists a derivation tree associated with $\atom$ and $\ruleset$ that contains an unbounded-path witness.
\end{propositionrep}

\begin{proof}
If there is no derivation tree having an unbounded-path witness, then the depth of all derivation trees is upper bounded by the number of sharing types. As derivation trees are of bounded arity, all derivation trees must be of bounded size.
 
If there is a derivation tree $\derivationtree$ having an unbounded-path witness $(\bag,\bag')$, we show that there are arbitrary large derivation trees. We do so by contradiction. Let us assume that $(\bag,\bag')$ is a UPW be two such bags such that $\bag'$ is of maximal depth among all such pairs and among all trees, which by hypothesis are of bounded size. Let $\bag_c$ be the child of $\bag$ that is on the shortest path from $\bag$ to $\bag'$ (possibly $\bag_c = \bag'$). By Proposition \ref{proposition-sharing-type-children-derivation-tree}, $\bag'$ has a child $\bag'_c$ that has the same sharing type as $\bag_c$. By Proposition \ref{proposition-sharing-type-children-derivation-tree}, $\bag'$ has a child $\bag'_c$ that has the same sharing type as $\bag_c$, either in the same tree, or in an extension of this tree, which is in contradiction with the fact that $\bag'$ was of maximal depth.  Hence there are arbitrary large derivation trees.
\end{proof}

The previous proposition yields a characterization of the existence of an infinite semi-oblivious derivation. At this point, one may notice that an infinite semi-oblivious derivation is not necessarily fair. However, from this infinite derivation one can always build a fair derivation by inserting missing triggers. Obviously, this operation has no effect on the termination of the semi-oblivious chase. More precaution will be required for the restricted chase.

One obtains an algorithm to decide termination of the semi-oblivious chase for a given set of rules: for each canonical instance, build a semi-oblivious derivation and the associated derivation tree by applying rules until a UPW  is created (in which case the answer is no) or all possible rule applications have been performed; if no instance has returned a negative answer, the answer is yes.  

\begin{corollary}
\label{corollary-semi-oblivious-finiteness-decidability}
The all-sequence termination problem for the semi-oblivious chase  on linear rules is decidable. 
\end{corollary}

\fussy
This algorithm can be modified to run in polynomial space (which is optimal \cite{DBLP:conf/pods/CalauttiGP15}), by guessing a canonical instance and a UPW of its derivation tree. 
\sloppy

\begin{propositionrep}
 The all-sequence termination problem for the semi-oblivious chase on linear rules is in \textsc{PSpace}.
\end{propositionrep}

\begin{proof}
 Let $\entailmenttree$ be a derivation tree of root the canonical instance $\{\atom\}$ that contains a UPW $(\bag,\bag')$, where the sharing type of both bags is $ST$. We show that there exists a semi-oblivious derivation of length at most exponential whose derivation tree has root $\{\atom\}$ and that contains a UPW $(\bag_s,\bag'_s)$ where the sharing type of both bags is $ST$. First, by Proposition \ref{proposition-sharing-type-children-derivation-tree}, we conclude that it is not necessary to have twice the same sharing type on the path from the root to $\bag'$ in the derivation tree. It is thus enough to show that to generate a child $\bag_c$ from its parent $\bag_p$, a derivation of length at most exponential is necessary. Let us consider the chase graph of the derivation generating $\atoms{\bag_c}$ from $\atoms{\bag_p}$. This chase graph can be assumed w.l.o.g. to be a path. It there are no pairs of atoms having the same sharing type on this path, then the derivation is of length at most exponential. Otherwise, we show that we can build a shorter semi-oblivious derivation that generates $\atoms{\bag_c}$. Let us thus assume that there is $\bag$ and $\bag'$ such that both have the same sharing type, and the terms of $\bag_p$ that appear in $\bag$ appear in the same position in $\bag'$, and that $\bag'$ is on the path from $\bag$ to $\bag_c$ in the chase graph. A derivation similar to that applicable after $\bag'$ is actually applicable to $\bag$, by Proposition \ref{proposition-sharing-type-children-derivation-tree}. A copy of $\bag_c$ under $\bag_p$ is thus generated by this derivation, which proves our claim. 
 
 We now describe the algorithm. We guess the canonical instance and the sharing type $ST$ of the UPW. We then check that there is a descendant (not necessarily a child) of that canonical instance that has sharing type $ST$. This can be done by guessing the shortest derivation creating a bag of sharing type $ST$. It is only necessary to remember the sharing type of the ``current'' bag, as we know that any bag created during a derivation is added as a descendant of the root. We then want to prove that a bag of sharing type $ST$ can have a (strict) descendant of sharing type $ST$. In contrast with the case of the root, a trigger applied below a bag $\bag$ does not necessarily create a bag that is as well below $\bag$ -- it could be added higher up in the tree. We thus have to remember the shared variables of $\bag$, and verify at each step that the shared variables of the currently considered bag are not a subset of them. This leads to a \textsc{PSpace} procedure.
\end{proof}

We now consider the restricted chase. To this aim, we call \emph{restricted derivation tree} associated with $\atom$ and $\ruleset$ a derivation tree associated with a restricted $\ruleset$-derivation from $\atom$. We first point out that Proposition \ref{proposition-sharing-type-children-derivation-tree} is not true anymore for a restricted derivation tree, as the order in which rules are applied matters. 

\begin{example}
Consider a restricted tree that contains bags $B$ and $B'$ with the same sharing type, labeled by atoms $q(t,u)$ and $q(v,w)$ respectively, where the second term is generated.
Consider the following rules (the same as in Example \ref{ex-intro}):\\
$\erule_1: q(x,y) \rightarrow \exists z ~q(y,z)$\\
$\erule_2: q(x,y) \rightarrow q(y,y)$\\
Assume $\bag$ has a child $\bag_c$ labeled by $q(u,z_0)$ obtained by an application of $\erule_1$, and $\bag'$ has a child $\bag'_1$ labeled by $q(w,w)$ obtained by an application of $\erule_2$. It is not possible to extend this tree by copying $\bag_c$ under $\bag'$. Indeed, the corresponding application of $\erule_1$ does not comply with the restricted chase criterion: it would produce an atom of the form $q(w,z_1)$ that folds into $q(w,w)$. 
\end{example}

We thus prove a weaker proposition by considering that $\bag'$ is a leaf in the restricted derivation tree.

\begin{proposition}
\label{proposition-sharing-type-restricted-derivation-tree}
  Let $\derivationtree$ be a prefix of a restricted derivation tree, $\bag$ and $\bag'$ be two bags of $\derivationtree$ such that $\bag \stequiv \bag'$ and \emph{$\bag'$ is a leaf}. Let $\bag_c$ be a child of $\bag$.
 Then: \emph{(a)} the tree obtained from $\derivationtree$ by adding the copy $\bag'_c$ of $\bag_c$ under $\bag'$ is a prefix of a restricted derivation tree, and \emph{(b)} it holds that $\bag_c \stequiv \bag'_c$.
\end{proposition}

\begin{proof} 
Let $S$ be the restricted derivation associated with $\derivationtree$. Let $S_c$ be the subsequence of $S$ that starts from $\bag$ and produces the strict descendants of $\bag$. Obviously, any rule application in  $S_c$ is performed on a descendant of $\bag$, hence we do not care about rule applications that produce bags that are not descendants of $\bag$. 
We prove the property by induction on the length of $S_c$. If $S_c$ is empty, the property holds with $\derivationtree_e = \derivationtree$. Assume the property is true for  $0 \leq |S_c| \leq k$. Let $|S_c| = k + 1$. By induction hypothesis, there is an extension $\derivationtree' $ of  $\derivationtree$ such that the subtree of $\bag$ restricted to the first $k$ elements of $S_c$ is `quasi-isomorphic' to the subtree rooted in $\bag'$ (via a bijective substitution defined by the natural mappings between sharing types, say $\phi$) . Let $(\erule,\match)$ be the last trigger of $S_c$, and assume it applies to a bag $\bag_d$. In $\derivationtree'$, there is a bag $\bag'_d = \phi(\bag_d)$. Hence, $\erule$ can be applied to $\bag'_d$ with the homomorphism $\phi \circ \match$. Any folding of the produced bag $\bag''$ to a bag in $\derivationtree'$ necessarily maps $\bag''$ to a bag in the subtree rooted in $\bag'_d$ (because $\bag'_d$ and $\bag''$ share a term generated in $\bag'_d$, that only occurs in the subtree rooted in $\bag'_d$ and remains invariant by the folding). Since $\bag_d$ and $\bag'_d$ have quasi-isomorphic subtrees, and $(\erule,\match)$ satisfies the restricted chase criterion, so does $(\erule,\phi \circ \match)$.
Furthermore, the quasi-isomorphism $\phi$ preserves the sharing types. Hence, $\bag'_d$ is added exactly like the bag produced by   $(\erule,\match)$.  
 We conclude that the property holds true at rank $k+1$. 
\end{proof}

The previous proposition allows us to obtain a variant of Proposition \ref{proposition-finiteness-derivation-tree} adapted to the restricted chase:  

\begin{propositionrep}
\label{proposition-finiteness-restricted-derivation-tree}
 There exists an arbitrary large restricted derivation tree associated with $\atom$ and $\ruleset$ if and only if there exists a restricted derivation tree associated with $\atom$ and $\ruleset$
  that contains an unbounded-path witness. 
\end{propositionrep}

\begin{proof}
If there is no restricted derivation tree with a UPW, then the size of any restricted derivation tree is bounded since a restricted derivation tree is a derivation tree. 
We prove the other direction by contradiction. Assume that the size of restricted derivation trees is bounded whereas the forbidden pattern occurs in some of them. Consider a restricted chase sequence $S$ with associated restricted derivation tree $\derivationtree$ that contains a UPW $(\bag,\bag')$ of maximal depth among all such pairs  and all trees, and such that $\bag'$ is a leaf (we can do the latter assumption since the prefix of any restricted derivation is a restricted derivation). 

Let $\bag_c$ be the child of $\bag$ that is on the shortest path from $\bag$ to $\bag'$ (possibly $\bag_c = \bag'$). By Proposition \ref{proposition-sharing-type-restricted-derivation-tree}, there is a restricted derivation tree that extends $\derivationtree$ such that $\bag'$ has a child $\bag'_c$ of the same sharing type as $\bag_c$, hence $(\bag_c, \bag'_c)$ is a UPW of depth strictly greater than $(\bag,\bag')$, which contradicts the hypothesis. 
\end{proof}

It is less obvious than in the case of the semi-oblivious chase that the existence of an infinite derivation entails the existence of an infinite \emph{fair} derivation. However, this property still holds:

\begin{propositionrep} 
For linear rules, every (infinite) non-terminating restricted derivation is a subsequence of a fair restricted derivation.
\end{propositionrep}

\begin{proof}
 Let $\derivation$ be a non-terminating restricted derivation. In particular, there exists a least one infinite branch in the associated derivation tree. Let us consider the following derivation: when the node $\bag_k$ of depth $k$ on this branch has been generated, complete the corresponding subsequence by trying to apply (i.e., while respecting the restricted criterion) all currently applicable triggers that add a bag a depth at most $k-1$. These additional rule applications cannot prevent the creation of any bag that is below $\bag_k$ in the derivation tree. Indeed, let $\atom_c$ be an atom possibly created by a rule application, whose bag would be attached as a child of a bag $\bag$; since $\atom_c$ shares a variable with $\atoms{\bag}$ that is generated in $B$, which thus only occurs in the subtree of $B$, the only possibility for $\atom_c$ to fold into the current instance, is to be mapped to an atom in the subtree of $\bag$. 
 By construction, any possible rule application will be performed or inhibited at some point, which implies that the derivation that we build in this fashion is fair.
\end{proof}
 
Similarly to Proposition \ref{proposition-finiteness-derivation-tree} for the semi-oblivious chase, Proposition \ref{proposition-finiteness-restricted-derivation-tree} provides an algorithm to decide termination of the restricted chase. 
The difference is that it is not sufficient to build a single derivation for a given canonical instance; instead, all possible restricted derivations from this instance have to be built (note that the associated restricted derivation trees are finite for the same reasons as before, and there is obviously a finite number of them). Hence, we obtain: 

\begin{corollary}
\label{corollary-restricted-finiteness-decidability}
The all-sequence termination problem for the restricted chase on linear rules is decidable. 
\end{corollary}

A rough analysis of the proposed algorithm provides a \textsc{co-N2ExpTime} upper-bound for the complexity of the problem, by guessing a derivation that is of length at most double exponential, and checking whether there is a UPW in the corresponding derivation tree.

\begin{figure}[t]

\begin{center}

\tikzset{every picture/.style={line width=0.75pt}} %set default line width to 0.75pt        

\begin{tikzpicture}[x=0.75pt,y=0.75pt,yscale=-1,xscale=1,thick,scale=0.6, every node/.style={scale=0.7}]

	\draw    (426, 38) circle [x radius= 55, y radius= 25]  ;

	\draw (426,38) node   {$p(x,y)$};

\draw (376,67) node   {$1$};

\draw    (426,63) -- (346,85) ;

\draw    (346, 110) circle [x radius= 55, y radius= 25]  ;

\draw (346,110) node   {$q(y)$};

		\draw (476,67) node   {$3$};

		\draw    (426,63) -- (506,85) ;

		\draw    (506, 110) circle [x radius= 55, y radius= 25]  ;

		\draw (506,110) node   {$r(y,x)$};

\draw (336,148) node   {$2$};

\draw    (346,135) -- (346,157) ;

\draw    (346, 182) circle [x radius= 55, y radius= 25]  ;

\draw (346,182) node   {$r(y,z_0)$};

\draw (336,220) node   {$4$};

\draw    (346,207) -- (346,229) ;

\draw    (346, 254) circle [x radius= 55, y radius= 25]  ;

\draw (346,254) node   {$p(z_0,z_1)$};

\draw (288,282) node   {$5$};

\draw [dotted]   (346,279) -- (266,299) ;

\draw (256,320) node   {$6$};

\draw [dotted]   (266,299) -- (266,329) ;

\draw (402,282) node   {$7$};

\draw [dotted]   (346,279) -- (426,299) ;

	\draw    (126, 38) circle [x radius= 55, y radius= 25]  ;

	\draw (126,38) node   {$p(x,y)$};

\draw (76,67) node   {$1$};

\draw    (126,63) -- (46,85) ;

\draw    (46, 110) circle [x radius= 55, y radius= 25]  ;

\draw (46,110) node   {$q(y)$};

		\draw (176,67) node   {$2$};

		\draw    (126,63) -- (206,85) ;

		\draw    (206, 110) circle [x radius= 55, y radius= 25]  ;

		\draw (206,110) node   {$r(y,x)$};

\end{tikzpicture}

\end{center}

\caption{Finite versus Infinite Derivation Tree for Example \ref{ex-bfs-stop}}

\label{figure-infinite-bf-derivation}

\end{figure}
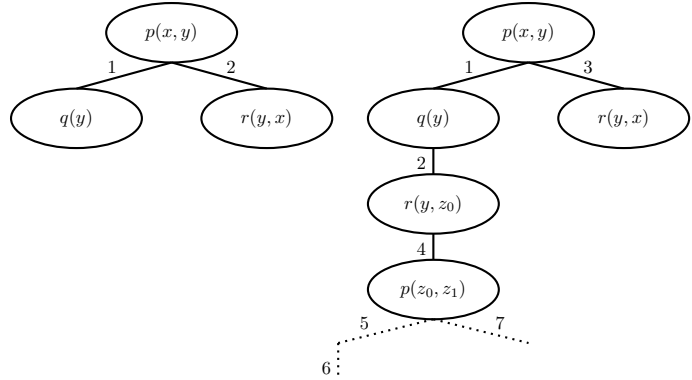

Importantly, the previous algorithm is naturally able to consider solely some type of restricted derivations, i.e., build only derivation trees associated with such derivations, which is of theoretical but also of practical interest. Indeed, implementations of the restricted chase often proceed by building \emph{breadth-first} sequences (which are intrinsically fair), or variants of these. As witnessed by the next example, the termination of all breadth-first sequences is a strictly weaker requirement than the termination of all fair sequences, in the sense that the restricted chase terminates on more sets of rules.  

\begin{example} \label{ex-bfs-stop}Consider the following set of rules:\\
$\erule_1 = p(x,y) \rightarrow q(y) \quad \quad \quad ~\erule_2 = p(x,y) \rightarrow r(y,x)$\\
$ \erule_3 = q(y) \rightarrow \exists z ~r(y,z) \quad \quad \erule_4 = r(x,y) \rightarrow \exists z ~p(y,z)$\\
All breadth-first restricted derivations terminate, whatever the initial instance is. Remark that every application of $\erule_1$ is followed by an application of $\erule_2$ in the same breadth-first step, which prevents the application of $\erule_3$. However, there is a fair restricted derivation that does not terminate (and this is even true for any instance). Indeed, an application of $\erule_2$ can always be delayed, so that it comes too late to prevent the application of $\erule_3$.  
See Figure \ref{figure-infinite-bf-derivation}: on the left, a finite derivation tree associated with a breadth-first derivation from instance $p(x,y)$; on the right, an infinite derivation tree associated with a (non breadth-first) fair infinite derivation from the same instance. The numbers on edges give the order in which bags are created.  

\end{example}

We now prove the decidability of the one-sequence termination problem, building on the same objects as before, but in a different way. 
Indeed, a (restricted) derivation tree $\derivationtree$ that contains a UPW $(B,B')$ is a witness of the existence of an infinite (restricted fair) derivation, but does not prove that \emph{every} (restricted fair) derivation that extends $\derivationtree$ is infinite. To decide, we will consider trees associated with a \emph{sharing type} instead of a type. A derivation tree associated with a sharing type $T$ has a root bag whose sharing type is $T$, and is built as for usual root bags, except that shared terms are taken into account, i.e., triggers $(\erule, \match)$ such that $\match(\fr{\erule}) \subseteq \shared{T}$ are simply ignored. The algorithm proceeds as follows:

\begin{enumerate}
\item For each sharing type $T$, generate all restricted derivations trees associated with $T$, stopping the construction of a tree when, for each leaf $B_L$, either there is no active trigger on $\atoms{B_L}$ or $B_L$ forms a UPW with one of its ancestors.    
\item Mark all the sharing types that have at least one associated tree without UPW. 
\item Propagate the marks until stability: if a sharing type $T$ has at least one tree for which all UPWs $(B,B')$ are such that the sharing type of $B$ is marked, then mark $T$. 
\item If all sharing types that correspond to instances (i.e., without shared terms) are marked, return \emph{yes}, otherwise return \emph{no}. 
\end{enumerate}

\begin{proposition}  The previous algorithm terminates and returns yes if and only if there is a terminating restricted sequence. 
\end{proposition}

\begin{proof}(Sketch)
Termination follows from the finiteness of the set of sharing types and the bound on the size of a tree. Concerning the correctness of the algorithm,
we show that a terminating restricted derivation cannot have a derivation tree that contains an unmarked UPW, i.e., whose associated sharing type is not marked. 
By contradiction: assume there is a terminating restricted derivation whose derivation tree contains an unmarked UPW; consider such an unmarked UPW $(B,B')$ such that $B'$ is of maximal depth in the tree. The subtree of $B'$ necessarily admits as prefix one of the restricted derivation trees associated with the sharing type of $B'$ built by the algorithm, otherwise the derivation would not be fair. Moreover, since the sharing type of $B'$ is not marked, this prefix contains an unmarked UPW. Hence, the tree contains an unmarked UPW $(B'',B''')$ with $B'''$ of depth strictly greater than the depth of $B'$, which contradicts the hypothesis. 
\end{proof}

\begin{corollary}
\label{corollary-one-sequence-finiteness-decidability}
The one-sequence termination problem for the restricted chase on linear rules is decidable. 
\end{corollary}

By guessing a terminating restricted derivation, which must be of size at most double exponential, and checking that the obtained instance is indeed a universal model, we obtain a \textsc{N2ExpTime} upper bound for the complexity of the one-sequence termination problem.

We conclude this section by noting that the previous Example \ref{ex-bfs-stop} may give the (wrong) intuition that, given a set of rules, it is sufficient to consider breadth-first sequences to decide if there exists a terminating sequence. 
The following example shows that it is not the case: here, no breadth-first sequence is terminating, while there exists a terminating sequence for the given instance.

\begin{example} \label{ex-bfs-non-stop-bis} 
Let $\ruleset=\{\erule_1,\erule_2,\erule_3\}$ with
$\erule_1 = p(x,y) \rightarrow  \exists z ~p(y,z) $, $\erule_2 = p(x,y) \rightarrow h(y)$, and $\erule_3= h(x) \rightarrow ~p(x,x)$.
In this case, for every instance, there is a terminating restricted chase sequence, where the application of $\erule_2$ and $\erule_3$ prevents the indefinite application of $\erule_1$. However, starting from $\instance = \{p(a,b)\}$, by applying rules in a breadth-first fashion 
one obtains a non-terminating restricted chase sequence, since
 $\erule_1$ and $\erule_2$ are always applied in parallel from the same atom, before applying $ \erule_3$.
\end{example}

As for the all-sequence termination problem, the algorithm may restrict the derivations of interest to specific kinds.

\section{Core Chase Termination}
We now consider the termination of the core chase of linear rules. Keeping the same approach, we prove that the finiteness of the core chase is equivalent to  the existence of a finite tree of bags whose set of atoms is a minimal universal model. We call this a \emph{(finite) complete core}. To bound the size of a complete core, we show that it cannot contain an unbounded-path witness. Note that in the binary case, it would be possible to work again on derivation trees, but this is not true anymore for arbitrary arity. Indeed, as shown in Example \ref{example-mlm}, there are linear sets of rules for which no derivation tree form a complete core (while it holds for binary rules). We thus introduce a more general tree structure, namely \emph{entailment trees}.

\begin{example}
 \label{example-mlm}
 Let us consider the following rules:
 \smallskip
 
 \begin{tabular}{ll}
   $s(x) \rightarrow \exists y \exists z ~p(y,z,x)$ & \hspace{1cm}  $p(y,z,x) \rightarrow \exists v ~q(y,v,x)$ \\
   $q(y,v,x) \rightarrow p(y,v,x)$ & ~ \\
 \end{tabular}

 \smallskip

Let $\instance = \{s(a)\}$. The first rule applications yield a derivation tree $\derivationtree$ which is a path of bags $B_0, B_1,B_2,B_3$ respectively labeled by the following atoms:\\
$s(a), p(y_0,z_0,a), q(y_0,v_0,a)$ and $p(y_0, v_0,a)$. $\derivationtree$ is represented on the left of Figure \ref{figure-example-mlm}. Let $A$ be this set of atoms. 
First, note that $ A$ is not a core: indeed it is equivalent to its strict subset $A'$ defined by $\{B_0, B_2, B_3\}$ with a homomorphism $\match$ that maps $\atoms{B_1}$ to 
$\atoms{B_3}$. Trivially, $A'$ is a core since it does not contain two atoms with the same predicate.  
Second, note that any further rule application on  $\derivationtree$ is redundant, i.e., generates a set of atoms equivalent to $A$ (and  $A'$).
Hence, $A'$ is a complete core, however there is no derivation tree that corresponds to it. There is even no \emph{prefix} of a derivation tree that corresponds to it (which ruins the alternative idea of building a prefix of a derivation  tree that would be associated with a complete core).  
In particular, note that $\{B_0, B_1, B_2\}$ is indeed a core, but it is not complete.

\end{example}

\begin{figure}[t]
\begin{center}
\tikzset{every picture/.style={line width=0.75pt}} %set default line width to 0.75pt        

\begin{tikzpicture}[x=0.75pt,y=0.75pt,yscale=-1,xscale=1,thick,scale=0.6, every node/.style={scale=0.7}]
%uncomment if require: \path (0,353); %set diagram left start at 0, and has height of 353

\draw    (126, 38) circle [x radius= 35, y radius= 20]  ;
\draw    (126.25, 104.5) circle [x radius= 55.25, y radius= 24.5]  ;
\draw    (126.25, 175.5) circle [x radius= 55.25, y radius= 24.5]  ;
\draw    (126.25, 246.5) circle [x radius= 55.25, y radius= 24.5]  ;
\draw    (126,58) -- (126,80) ;

\draw    (126,129) -- (126,151) ;

\draw    (126,200) -- (126,222) ;

\draw    (68.1,229.09) .. controls (6.16,198.84) and (14.56,141.8) .. (70,122) ;

\draw [shift={(70,230)}, rotate = 205.11] [color={rgb, 255:red, 0; green, 0; blue, 0 } ][line width=0.75]    (10.93,-3.29) .. controls (6.95,-1.4) and (3.31,-0.3) .. (0,0) .. controls (3.31,0.3) and (6.95,1.4) .. (10.93,3.29)   ;
\draw    (375, 76) circle [x radius= 35, y radius= 20]  ;
\draw    (375.25, 143.5) circle [x radius= 55.25, y radius= 24.5]  ;
\draw    (375, 214.5) circle [x radius= 55.25, y radius= 24.5]  ;
\draw    (375,96) -- (375,118) ;

\draw    (375,168) -- (375,190) ;

\draw (126,37) node   {$s( a)$};
\draw (124,105) node   {$p( y_{0} ,z_{0} ,a)$};
\draw (125,176) node   {$q( y_{0} ,v_{0} ,a)$};
\draw (126,247) node   {$p( y_{0} ,v_{0} ,a)$};
\draw (76,38) node   {$B_{0}$};
\draw (56,105) node   {$B_{1}$};
\draw (54,175) node   {$B_{2}$};
\draw (57,250) node   {$B_{3}$};
\draw (117,300) node  [align=left] {Derivation Tree};
\draw (364,299) node  [align=left] {An Entailment Tree};
\draw (15,170) node   {$\varphi $};
\draw (372,75) node   {$s( a)$};
\draw (374,144) node   {$q( y_{0} ,v_{0} ,a)$};
\draw (375,215) node   {$p( y_{0} ,v_{0} ,a)$};
\draw (322,76) node   {$B_{0}$};
\draw (303,143) node   {$B_{2}$};
\draw (306,218) node   {$B_{3}$};

\end{tikzpicture}
\end{center}
\caption{Derivation tree and entailment tree for Example \ref{example-mlm}}
\label{figure-example-mlm}
\end{figure}
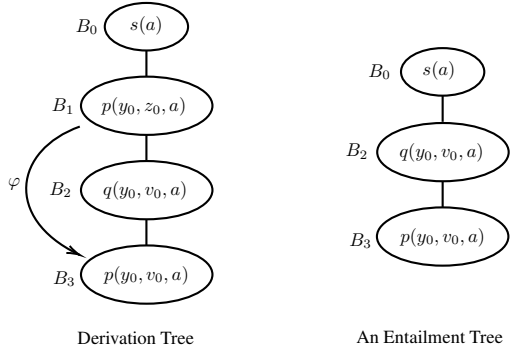

In the following definition of entailment tree, we use the notation $\atom_1 \rightarrow \atom_2$, where $\atom_i$ is an atom, to denote the rule $\forall X (\atom_1 \rightarrow \exists Y ~\atom_2)$ 
with $X = \vars{\atom_1}$ and $Y = \vars{\atom_2}\setminus X$.

\begin{definition}[Entailment Tree]
\label{definition-entailment-tree}
 An \emph{entailment tree} associated with $\atom$ and $\ruleset$ is a tree of bags $\entailmenttree$ such that:
  \begin{enumerate}
  \item  $\bag_r$, the root of $\entailmenttree$, is such that $\ruleset \models \atom \rightarrow \atoms{\bag_r}$ and $\ruleset \models \atoms{\bag_r} \rightarrow \atom$; 
  \item \sloppypar{For any bag $\bag_c$ child of a node $\bag$, the following holds: \emph{(i)} $\terms{\bag_c}\cap \generated{\bag} \neq \emptyset$
  \emph{(ii)} The terms in $\generated{\bag_c}$ are variables that do not occur outside the subtree of $\entailmenttree$ rooted in $\bag_c$  \emph{(iii)} $\ruleset \models \atoms{\bag}  \rightarrow \atoms{\bag_c}$.}
 \item there is no pair of twins.
 \end{enumerate}
\end{definition}

Note that $\atom$ is not necessarily the root of the entailment tree, as it may not belong to the result of the core chase on $\atom$ (hence Point 1). 

First note that an entailment tree is independent from any derivation. The main difference with a derivation tree is that it employs a more general parent-child relationship, that relies on entailment rather than on rule application, hence the name entailment tree. Intuitively, with respect to a derivation tree, one is allowed to move a bag $B$ higher in the tree, provided that it contains at least one term generated in its new parent $B_p$; then, the terms of $B$ that are not shared with $B_p$ are freshly renamed. 
Finally, since the problem of whether an atom is entailed by a linear existential rule knowledge base is decidable (precisely \textsc{PSpace}-complete \cite{DBLP:books/sp/virgilio09/CaliGL09}), one can actually generate all non-twin children of a bag and keep a tree with bounded arity.

Derivation trees are entailment trees, but not necessarily conversely. 
A crucial distinction between these two structures is the following statement, which does not hold for derivation trees, as illustrated by Example \ref{example-mlm}. 

\begin{propositionrep}
\label{proposition-entailment-core}
 If the core chase associated with $\atom$ and $\ruleset$ is finite, then there exists an entailment tree $\entailmenttree$ such that the set of atoms associated with $\entailmenttree$ is a complete core.
\end{propositionrep}

\medskip
\noindent \emph{Example \ref{example-mlm} (cont'd).} The tree defined by the path of bags $B_0$, $B_2$, $B_3$ is an entailment tree, represented on the right of Figure \ref{figure-example-mlm}, which defines a complete core.

\begin{proof}
Let $\derivationtree$ be the derivation tree associated with a derivation containing a core $C$ of $\chase{\atom}{\ruleset}$. Let $\varphi$ be an idempotent homomorphism from the atoms of $\derivationtree$ to $C$. We assign to each bag $\bag$ of $\derivationtree$ a set of trees $\{T_1,\ldots,T_{n_B}\}$ such that:
\begin{enumerate}
 \item each tree contains only elements of $C$;
 \item the forest assigned to $\bag$ contains exactly once the elements of $C$ appearing in the subtree rooted in $\bag$;
 \item for each pair $(\bag_p,\bag_c)$ of bags in some $T_i$ such that $\bag_p$ is a parent of $\bag_c$, $\ruleset \models \atoms{\bag_p} \rightarrow \atoms{\bag_c}$;
 \item each $T_i$ is a decomposition tree;
 \item for each $T_i$, the root of $T_i$ contains all the terms that belong both to $T_i$ and to $C \setminus T_i$;
 \item each term $\term$ belonging to distinct $T_i$ and $T_j$ of the forest assigned with a bag $\bag$ also belongs to the parent of $\bag$.	
\end{enumerate}

Moreover, we will show that if $\varphi(\bag)$ is a descendant of $\bag$ (including $\bag$) in $\derivationtree$, then its associated forest is a tree.

\begin{itemize}
\item if $\bag$ is a leaf, we consider two cases:
 \begin{itemize}
   \item $\bag$ belongs to the core: we assign it a single tree, containing only a root being itself. All conditions are trivial.
   \item $\bag$ does not belong to the core: we assign it an empty forest, and all conditions are trivial.
 \end{itemize}
\item if $\bag$ is an internal node, let $\{T_1,\ldots,T_n\}$ be the union of the forests assigned to the children of $\bag$. We distinguish three cases:
  \begin{itemize}
   \item $\bag$ is in the core: we assign to $\bag$ the tree $T$ containing $\bag$ as root, and having as children the roots of $\{T_1,\ldots,T_n\}$.
    \begin{itemize}
     \item 1. 2.: holds by induction assumption, the fact that different $T_i$'s cover disjoint subtrees of $\derivationtree$, and the fact that $\bag$ belongs to the core
     \item 3.: it is enough to check this for the pairs (root of $T$, root of $T_i$). The root of $T$ is an ancestor of root of $T_i$ in $\derivationtree$, hence $\ruleset \models \atoms{\mathrm{root}(T)}\rightarrow \atoms{\bag_i}$, where $\bag_i$ is the root of $T_i$
     \item 4. if $\term$ appears in $T$ but in no $T_i$, it appears only in $\bag$ and the connectivity of the substructure containing $\term$ holds. If it belongs to some $T_i$ and to $C \setminus T_i$, it must belong to the root of $T_i$ by assumption 6.. If $\term$ belongs to $C \setminus T$, it belongs to $\bag$ by connectivity of $\derivationtree$. If $\term$ belongs to another $T_j$, we distinguish two cases: $T_j$ is in the same forest as $T_i$, and then by induction assumption 7. on the child of $\bag$ to which this forest is associated, $\term$ belongs to $\bag$. Or $T_j$ is in the forest of another child of $\bag$, and then by connectivity property for $\term$, it belongs to $\bag$. Hence the connectivity property for $\term$ in $T$ is fulfilled.
     %\item 5. By 3. and 4.
     \item 5. By connectivity of $\derivationtree$, as $\bag$ is the root of $T$
     \item 6. true as there is only one tree
    \end{itemize}
   \item $\varphi(\bag) \not = \bag$ but is a descendant of $\bag$. By induction assumption 2., there exists exactly one tree among the trees associated with children of $\bag$ containing $\varphi(\bag)$. Let assume w.l.o.g that it is $T_1$, of root $\bag_1$. We build the following tree $T$: for all $T_i \not = T_1$, we add to $\bag_1$ a subtree by putting the root of $T_i$ under $\bag_1$.
   \begin{itemize}
   \item 1. No added elements, hence by induction assumption 1.
   \item 2. No added elements, hence by induction assumption 2.
   \item 3. To check for pairs ($\bag_1$,$\bag_i$), where $\bag_i$ is the root of $T_i$. $\ruleset \models \atoms{\bag_1} \rightarrow \atoms{\varphi(\bag)}$, as $\varphi(\bag)$ is a descendant of $\bag_1$ in $T_1$. Moreover, $\ruleset \models \atoms{\varphi(\bag)} \rightarrow \atoms{\bag_i}$, as $\varphi(\bag)$ is more specific than $\bag$, and $\varphi$ is the identity on shared terms. 
   \item 4. for all term $\term$ appearing in a single tree, the connectivity property holds by induction assumption 4.. Let $\term$ appearing in two trees. $\term$ appears in the roots of both tree by $6.$, and must appear in $\bag$ by connectivity of $\derivationtree$, hence in $\varphi(B)$, and hence in $\bag_1$ (by 6.). As $\bag_1$ and the roots of both trees are neighbor, this proves the result.
   %\item 5. by 3 and 4.
   \item 5. let $\term$ belonging to $T$ and to $C \setminus T$. By connectivity of $\derivationtree$, $\term$ belongs to $\bag$, hence to $\varphi(\bag)$ (because $\varphi(t) =t$). As $\term$ belongs both to $T_1$ and to $C \setminus T_1$, $\term$ belongs to $\bag_1$, and hence to the root of the assigned tree.
   \item 6. true as there is only one tree.
   \end{itemize}
   \item $\varphi(\bag)$ is not a descendant of $\bag$. We assign to $\bag$ the union of the forests associated to its children.
   \begin{itemize}
   \item 1.-5 By induction assumption
   \item 6. let $\term$ belonging to two trees $T_1$ and $T_2$. If $T_1$ and $T_2$ come from forest associated to two different children, $\term$ belongs to $\bag$ by connectivity of $\derivationtree$. If $T_1$ and $T_2$ come from the same forest, $\term$ belongs to $\bag$ by induction assumption 7. Then $\term$ belongs to $\bag$. As $\term$ is in $C$, $\term$ belongs to $\varphi(\bag)$. By connectivity of $\derivationtree$, it belongs to the parent of $\bag$, because that parent is on the path from $\bag$ to $\varphi(\bag)$, which proves 6.
   \end{itemize}
  \end{itemize}
\end{itemize}
Finally, we check that the following property is satisfied: for any bag $\bag$, if $\bag$ is in the core, then a single tree with root $\bag$ is assigned to it. If $\atom$ is in the core, we have built such a tree. It remains to obtain an entailment tree: for that, we have to bring up nodes at the highest level with respect to shared terms. We may also have to say something about 'generated' if it still appear in the definition of an entailment tree.
\end{proof}

\medskip
Differently from the semi-oblivious case, we cannot conclude that the chase does not terminate as soon as a UPW is built, because the associated atoms may later be mapped to other atoms, which would remove the UPW.  Instead, starting from the initial bag, we recursively add bags that do not generate a UPW
(for instance, we can recursively add all such non-twin children to a leaf). Once the process terminates (the non-twin condition and the absence of UPW ensure that it does), we check that the obtained set of atoms $C$ is complete (i.e., is a model of the KB): for that, it suffices to perform each possible rule application on $C$ and check if the resulting set of atoms is equivalent to $C$.  See Algorithm \ref{algorithm-core-chase}. The set $C$ may not be a core, but it is complete iff it contains a complete core.

We now focus on the key properties of entailment trees associated with complete cores. We first introduce the notion of \emph{redundant bags}, which captures some cases of bags that cannot appear in a finite core. As witnessed by Example \ref{example-mlm}, this is not a characterization: $B_1$ is not redundant (according to next Definition \ref{definition-redundancy}), but cannot belong to a complete core.

\begin{definition}[Redundancy]
\label{definition-redundancy}
Given an entailment tree, a bag $\bag_c$ child of $\bag$ is redundant if there exists an atom $\beta$ (that may not belong to the tree) with \emph{(i)} $\ruleset \models \atoms{\bag} \rightarrow \beta$; \emph{(ii)} there is a homomorphism from $\atoms{\bag_c}$ to $\beta$ that is the identity on $\shared{\bag_c}$ \emph{(iii)} $|\terms{\beta} \setminus \terms{\bag}| < |\terms{\bag_c} \setminus \terms{\bag}|$.
\end{definition}

Note that $\bag_c$ may be redundant even if the ``cause'' for redundancy, i.e., $\beta$, is not in the tree yet. 

The role of this notion in the proofs is as follows: we show that if a complete entailment  tree contains a UPW then it contains a redundant bag, and that a complete core cannot contain a redundant bag, hence a UPW. To prove this, we rely on next Proposition \ref{proposition-swissknife-bag-copy}, which is the counterpart for entailment trees of Proposition \ref{proposition-sharing-type-children-derivation-tree}: performing a bag copy from an entailment tree results in an entailment tree (the notion of prefix is not needed, since a prefix of an entailment tree is an entailment tree) and keeps the properties of the copied bag.  
\begin{propositionrep} 
\label{proposition-swissknife-bag-copy}
Let $\bag$ be a bag of an entailment tree $\entailmenttree$, $\bag'$ be a bag of an entailment tree $\entailmenttree'$ such that $\bag \stequiv \bag'$. Let $\bag_c$ be a child of $\bag$ and $\bag_c'$ be a copy of $\bag_c$ under $\bag'$. Let $\entailmenttree''$ be the extension of $\entailmenttree'$ where $\bag_c'$ is added as a child of $\bag'$. Then \emph{(i)} $\entailmenttree''$ is an entailment tree; \emph{(ii)} $\bag_c$ and $\bag_c'$ have the same sharing type;
\emph{(iii)} $\bag_c'$ is redundant if and only if $\bag_c$ is redundant.
\end{propositionrep}

In light of this, the copy of a bag can be naturally extended to the copy of the whole subtree rooted in a bag, which is crucial element in the proof of next Proposition~\ref{proposition-uc-excluded-main-text}:

\begin{toappendix}
Another important property of entailment trees (which is also satisfied by derivation trees) is that its structure provides information on where a bag may be mapped by $\varphi$ if its parent is left invariant by $\varphi$. 
\begin{lemma}
\label{lemma-locality}
 Let $\entailmenttree$ be an entailment tree. Let $\varphi$ be a homomorphism from the atoms of $\entailmenttree$ to themselves. Let $\bag_p$ such that $\varphi_{\mid \terms{\bag_p}}$ is the identity. Let $\bag_c$ be a child of $\bag_p$. Then $\varphi(\bag_c)$ is in the subtree rooted in $\varphi(\bag_p) = \bag_p$.
\end{lemma}
\begin{proof}
 $\bag_c$ is a child of $\bag_p$ thus there exists at least one term generated in $\bag_p$ that is a term of $\bag_c$. As $\varphi$ is the identity on $\bag_p$, that term belongs as well to $\varphi(\atoms{\bag_c})$. Thus $\varphi(\atoms{\bag_c})$ should also be in a bag that is in the subtree rooted in $\bag_p$.
\end{proof}
\end{toappendix}

\begin{proposition}
\label{proposition-uc-excluded-main-text}
 A complete core cannot contain \emph{(i)} a redundant bag \emph{(ii)} an unbounded-path witness.
\end{proposition}

\begin{toappendix}
\begin{proposition}
\label{proposition-strong-redundancy}
 A complete core cannot contain a redundant bag.
\end{proposition}

\begin{proof}
 Let $\entailmenttree$ be a complete entailment tree, and let $\hat{\bag}$ be a bag that is redundant. 
 %We perform as in the proof of Proposition \ref{proposition-child-redundancy} and 
 We prove that there exists a non-injective endomorphism of $\entailmenttree$, showing that $\entailmenttree$ cannot be a core. For any entailment tree $\entailmenttree_p$ that is a prefix of $\entailmenttree$, we build $\entailmenttree'_{p}$ and a mapping $\varphi$ from the terms of $\entailmenttree_p$ to the terms of  $\entailmenttree'_{p}$ as follows:
 \begin{itemize}
  \item for any prefix of $\entailmenttree_p$ that does not contain $\hat{\bag}$, we define $\entailmenttree_{p}' = \entailmenttree_p$ and $\varphi$ the identity
  \item for the prefix that contains all nodes of $\entailmenttree$, including $\hat{\bag}$, except the descendants of $\hat{\bag}$, we define $\entailmenttree'_p$ as $\entailmenttree_p$ to which we add a leaf (if necessary) to the parent of $\hat{\bag}$ in $\entailmenttree$, which is the witness of the redundancy of $\hat{\bag}$. We define $\varphi$ as the identity on any term that is not generated in $\hat{\bag}$, and as its image by the $\varphi$ witnessing the redundancy pattern on terms generated in $\hat{\bag}$.
  \item if we have defined $\entailmenttree'_{p}$ for $\entailmenttree_p$, and we want to define $\varphi$ for $\entailmenttree'_n$ for $\entailmenttree_n$ which is $\entailmenttree_p$ to which a leaf $\bag_d$ has been added, we add where it belongs the bag $\varphi(\bag_d)$, where we extend $\varphi$ to term generated in $\bag_d$ by choosing fresh images. 
 \end{itemize}

 By construction, $\entailmenttree'$ is an entailment tree, and $\varphi$ is a homomorphism from $\entailmenttree$ to $\entailmenttree'$. Moreover, $\varphi$ is not injective: indeed, as $\hat{\bag}$ is redundant, $\varphi$ is not injective on the terms of $\hat{\bag}$. 
 
 As $\entailmenttree$ is complete, there exists a homomorphism from $\entailmenttree'$ to $\entailmenttree$. Hence the composition of the two homomorphisms is a homomorphism from $\entailmenttree$ to itself, which is not injective, as $\varphi$ is not. Hence $\entailmenttree$ is not a core.
\end{proof}
\end{toappendix}

\begin{toappendix}
\begin{proposition}
\label{proposition-uc-excluded} A complete core cannot contain any unbounded-path witness.
\end{proposition}

\medskip
\begin{proof} 
We prove the result by contradiction. Let us assume that $\entailmenttree$ is a complete core containing an unbounded-path witness $(\bag,\bag')$. Let us choose $(\bag,\bag')$ such that $\bag'$ is of maximal depth with respect to its branch, that is, there is no unbounded-path witness $(\bag''',\bag'')$ with $\bag''$ a strict descendant of $\bag'$.

\medskip
Let $\bag_c$ be the child of $\bag$ on the path from $\bag$ to $\bag'$. % and $\bag'_c$ be a copy of $\bag_c$ under $\bag'$.
Let us denote by $\entailmenttree_{\bag_c}$ the subtree of $\entailmenttree$ which is rooted at $\bag_c$ and by $\entailmenttree_{\bag_c'}$ a copy of $\entailmenttree_{\bag_c}$ under $B'$ whose root is $\bag_c'$. 
Then, let 
$\entailmenttree'$ be the extension of $\entailmenttree$ where $\entailmenttree_{\bag_c'}$ is added as a child of $\bag'$.
We want to show that there exists a  bag $\bag_r'$ child of $\bag'$ and a mapping  from $\entailmenttree_{\bag_c'}$ into $\entailmenttree_{\bag_r'}$, which is  the identity on the terms of $\entailmenttree$. More precisely, we want to show that for each $\bag_d'$ descendant of $ \bag_c'$ the following properties hold.
\begin{enumerate}
\item the image of $\bag_d'$ belongs to $\entailmenttree_{\bag_r'}$
\item the image of a term generated in $\bag_d'$ is a term generated in a bag of $\entailmenttree_{\bag_r'}$
\end{enumerate}
We do so by induction on $k$ the distance between $\bag_d'$ and $\bag_c'$ in $\entailmenttree$.
\begin{itemize}
\item If $k=0$ then $\bag_d'=\bag_c'$. 
 Because $\entailmenttree$ is a complete core, there exists a homomorphism from the atoms of $\mathcal{T'} $ to those of $ \mathcal{T}$ which is the identity on the terms of $\entailmenttree$. 

 We show that the image of $\bag'_c$ is a strict descendant of a child of $\bag'$. 
Note first that no child of $\bag'$ in $\entailmenttree$ can be a safe renaming of $\bag'_c$. Indeed, by Proposition \ref{proposition-swissknife-bag-copy}, $\bag_c$ and $\bag_c'$ have the same sharing type and therefore  $\bag_c'$ (as well as any safe renaming of its generated terms) cannot be a child of $\bag'$ because the couple $(\bag_c,\bag_c')$ would form 
an unbounded-path witness with $\bag_c'$ strictly deeper than $\bag'$. 
 Then, if $\bag_c'$ maps to $\bag'$ then the couple $(\bag',\bag_c')$ is  redundant 
and therefore also $(\bag,\bag_c)$ is redundant, by Proposition \ref{proposition-swissknife-bag-copy},  which in turn implies that $\entailmenttree$ is not a core,  because of Proposition \ref{proposition-strong-redundancy}.
Finally, if $\bag_c'$ maps to any child of $\bag'$ then it does so by specializing the sharing type of $\bag_c'$ (as we showed that no safe renaming of $\bag_c'$ can be a child of $\bag'$), which means that  $\bag_c'$ is redundant. Therefore, by Proposition \ref{proposition-swissknife-bag-copy}, $\bag_c$ is also redundant and hence, by Proposition \ref{proposition-strong-redundancy}, $\entailmenttree$ is not a core. This proves that the image of $\bag_c'$ is a strict descendant of some $\bag_r'$ child of $\bag'$.

\smallskip
Now, to prove the second point, let  $t$ be a term generated in $\bag'_c$ and $t'$ its image.
It is easy to see that for entailment trees any term that belongs to two bags in ancestor-descendant relationship also belongs to the bags on the shortest path between them.
Therefore, if $t'$ is generated by a strict ancestor of $\bag_r'$ then  $t'$ belongs to the terms of $\bag'$.  This means that starting from the sharing type of $\bag_c'$ one can build a strictly more specific sharing type where the position corresponding to the generated term $t$ becomes shared with $\bag'$. From this one can find a node $\bag_c''$  which is strictly more specific than $\bag_c'$ and that can be added as a child of $\bag'$. This means that $\bag_c'$ is redundant and by Proposition \ref{proposition-swissknife-bag-copy} also $\bag_c$ is redundant, so $\entailmenttree$ is not a core.

\item
Let us assume that both properties hold for all bags at distance $k$ from $\bag_c'$.
We want to show that they still holds for the bags at distance $k+1$.

Let $\bag_d'$ be a node at distance $k+1$ from $\bag_c'$ whose parent is $\bag'_\delta$.
By definition, $\bag_d'$  contains a term generated by $B'_\delta$ and, by induction, we know that the image of this term is generated in a bag of $\entailmenttree_{\bag_r'}$. Thus, it follows that the image of $\bag_d'$ belongs to  $\entailmenttree_{\bag_r'}$ as required by the first point.

 \smallskip
For the second point we reason by contradiction and show that when the property does not hold then 
$\entailmenttree$  admits a non-injective endomorphism 
and thus it is not a core.
We proceed with the following construction.
Let $\entailmenttree_{\bag_r}$ be a copy of $\entailmenttree_{\bag_r'}$ under $\bag$ and $\entailmenttree''$ the extension of $\entailmenttree$ where $\entailmenttree_{\bag_r}$ is added as a child of $\bag$. 
We know by induction that there exists a homomorphism from $\entailmenttree'$ to $\entailmenttree$ mapping all nodes at distance ${k+1}$ from $\bag'_c$ to the subtree rooted at $\bag'_r$.
From this, we can conclude that there exists a homomorphism from $\entailmenttree_{(k+1)}$ to $\entailmenttree''$, where  $\entailmenttree_{(k+1)}$ is the prefix of $\entailmenttree$ which includes all nodes of $\entailmenttree$ except for the descendants of $\bag_c$  that are at distance strictly greater than $k+1$ from it. 
Now, we further extend $\entailmenttree''$ by adding an image for all nodes which are at distance strictly greater than $k+1$ from  $\bag_c$  thereby obtaining a new entailment tree $\entailmenttree'''$. 
It follows that  $\entailmenttree$ can be mapped to $\entailmenttree'''$. 
Beside, since $\entailmenttree$ is complete there exists an homomorphism from $\entailmenttree'''$ to $\entailmenttree$.
So, by composing these two homomorphisms we get a homomorphism from $\entailmenttree$ to $\entailmenttree$.

We show that the homomorphism from $\entailmenttree$ to $\entailmenttree'''$ is non-injective.
Recall  that to construct $\entailmenttree'$ the whole subtree rooted at $\bag_c'$  has been copied from the subtree rooted at $\bag_c$. Let us denote by $\bag_d$ the node at distance $k+1$ from $\bag_c$ from which $\bag_d'$ has been copied under $\bag_\delta'$. 
Let  $t$ be a term generated at position $i$ in $\bag_d'$. If its image was generated by a strict ancestor of $\bag_r'$ then this would also belong to the terms of $\bag'$.
By Proposition \ref{proposition-swissknife-bag-copy}, $\bag_d$ and $\bag_d'$ have the same sharing types, hence  the mapping from $\entailmenttree_{(k+1)}$ to $\entailmenttree''$ (and thus that from $\entailmenttree $ to $\entailmenttree'''$) maps the generated term  at position $i$ of $\bag_d$, we call $s$, to a distinct term in $\bag$, we call $s'$. Moreover, the homomorphism is the identity on $s'$.
Therefore, the homomorphism from $\entailmenttree$ to $\entailmenttree'''$ is non-injective as both $s'$ and $s$ have the same image.

\end{itemize}
To finish the proof,  we proceed with the following construction.
Let $\entailmenttree^*$ be an entailment tree derived from $\entailmenttree$ where $i)$ the whole subtree rooted at $\bag_r'$ has been copied under $\bag$  and $ii)$ the subtree rooted at $\bag_c$ has been removed.
Note that  $\entailmenttree^*$ is of size strictly smaller than that of  $\entailmenttree$ because we added a bag for each descendant node of $\bag'_r$, which is a strict descendant of bag $\bag_c$, and that this last one has been removed. Now, because $\entailmenttree_{\bag_c'}$ maps to $\entailmenttree_{\bag_r'}$ it follows that $\entailmenttree_{\bag_c}$ maps to $\entailmenttree_{\bag_r}$ and by extending this homomorphism to the identity on all other terms we get that $\entailmenttree$ can be mapped to $\entailmenttree^*$. Hence, $\entailmenttree$ is not a core. 
\end{proof}
\end{toappendix}

\begin{corollary}
\label{corollary-core-finiteness-decidability}
The all-sequence termination problem for the core chase on linear rules is decidable. 
\end{corollary}

\begin{algorithm}
\SetKwInOut{Input}{Input}\SetKwInOut{Output}{Output}
 \Input{A set of linear rules}
 \Output{\texttt{true} if and only if the core chase terminates on all instances}
 \For{each canonical atom $\atom$}
 {
 Let $\derivationtree$ be the entailment tree restricted to $\atom$;\\
 \While{a bag $\bag$ can be added to $\derivationtree$ respecting twin-free entailment tree condition and without creating a UPW}{
add $ \bag$ to $\derivationtree$
 } 
 \If{there is a rule $\erule$ applicable to $\treeatoms{\derivationtree}$ through $\match$ s.t. $\treeatoms{\derivationtree} \not \models \treeatoms{\derivationtree} \cup \matchsafe(\head{\erule})$}
   {\Return \texttt{false}}
 }
\Return \texttt{true}
 \caption{Deciding core chase termination}
 \label{algorithm-core-chase}
\end{algorithm}

A rough complexity analysis of this algorithm yields a \textsc{2ExpTime} upper bound for the termination problem. Indeed, the exponential number of (sharing) types yields a bound on the number of canonical instances to be checked, the arity of the tree, as well as the length of a path without UPW in the tree, and each edge can be generated with a call to a \textsc{PSpace} oracle.

\section{Concluding remarks} 

We have shown the decidability of chase termination over linear rules for three main chase variants (semi-oblivious, restricted, core) following a novel  approach based on derivation trees, and their generalization to entailment trees, and a single notion of forbidden pattern. As far as we know, these are the first decidability results for the restricted chase, on both versions of the termination problem (i.e., \emph{all sequence} and \emph{one sequence} termination). The simplicity of the structures and algorithms make them subject to implementation. 

We leave for future work the study of the precise complexity of the termination problems. A straightforward analysis of the complexity of the algorithms that decide the termination of the restricted and core chases yields upper bounds, however we believe that a finer analysis of the properties of sharing types would provide tighter upper bounds.  Future work also includes the extension of the results to more complex classes of existential rules: linear rules with a complex head, which is relevant for the termination of the restricted and core chases, and more expressive classes from the guarded family. Derivation trees were precisely defined to represent derivations with guarded rules and their extensions (i.e., greedy bounded treewidth sets), hence they seem to be a promising tool to study chase termination on that family.

 \bibliography{bib}
 \bibliographystyle{alpha} 

\end{document}